\documentclass{article}
\pdfoutput=1
\usepackage{graphicx} 
\usepackage{dcolumn}
\usepackage{bm}
\usepackage{xcolor}
\usepackage{physics}
\usepackage{float}
\usepackage{amsfonts}
\usepackage{amsmath}
\usepackage{amsthm}
\usepackage{mathtools}
\usepackage{tcolorbox}
\usepackage{amssymb}
\usepackage{fullpage}
\usepackage{hyperref}
\usepackage{algorithm}
\usepackage[noend]{algpseudocode}

\newcommand{\polyn}{\mathrm{poly(n)}}

\definecolor{scmcolor}{RGB}{0,0,255}

\definecolor{scmcolorRed}{RGB}{255,0,0}

\newtheorem{theorem}{Theorem}
\newtheorem*{theorem*}{Theorem}
\newtheorem{definition}[theorem]{Definition}
\newtheorem{corollary}[theorem]{Corollary}
\newtheorem*{corollary*}{Corollary}

\newtheorem{lemma}[theorem]{Lemma}

\title{On Bounded Advice Classes}

\author{Simon C. Marshall \footnote{s.c.marshall@liacs.leidenuniv.nl}
\and Casper Gyurik\vspace{0.5cm}\\Universiteit Leiden
\and Vedran Dunjko}
\date{}

\begin{document}
\maketitle

\begin{abstract}
    Advice classes in computational complexity have frequently been used to model real-world scenarios encountered in cryptography, quantum computing and machine learning, where some computational task may be broken down into a preprocessing and deployment phase, each associated with a different complexity.
    However, in these scenarios, the advice given by the preprocessing phase must still be generated by some (albeit more powerful) bounded machine, which is not the case in conventional advice classes. 
    To better model these cases we develop `bounded advice classes', where a more powerful Turing machine generates advice for another, less powerful, Turing machine.
    We then focus on the question of when various classes generate useful advice, to answer this we connect bounded advice to unary languages. This connection allows us to state various conditional and unconditional results on the utility of advice generated by $\mathsf{EXP}$, $\mathsf{NP}$, $\mathsf{BQP}$, $\mathsf{PSPACE}$, and more. 
    We study the relations between bounded advice classes, quantum bounded advice classes, and randomised bounded advice. We also examine how each of these concepts interact with recently introduced classes, like $\mathsf{BPP/samp}$.
    Our results also improve the state of the art in existing research on the complexity of advice functions. 
\end{abstract}


\section{Introduction}
A common concept in complexity theory is `advice', where a Turing machine is provided with an advice string alongside its input. This string can provide any conceivable advice on how to solve problems up to size $n$ but must be the same advice string for all inputs of a given size.
As this advice string could be anything, advice classes are often much more powerful than their unadvised counterparts, containing problems which are otherwise undecidable.
This powerful advice string allows advice classes to act as useful models for problems with long/expensive preprocessing phases (modelled as advice) followed by faster/cheaper processing phases.
However, in many of these situations, the advice is produced by a more powerful Turing machine, but not an infinitely more powerful Turing machine, as would be the case in standard advice classes.
In such situations, it would be more appropriate to study a form of `bounded advice', where we ask what a given Turing machine of some complexity can accomplish when given advice generated by a more powerful, but not unbounded, Turing machine. 
For many cases where we might want to use advice classes, this bounded advice captures a much tighter idea of what we are trying to model:

\begin{itemize}
    \item Cryptography, when an adversary may be able to expend immense resources if the information learnt allows them to break an encryption scheme for a much smaller cost during run-time
    \cite{adrian2015imperfect, aaronson2015NSA}. 
    
    \item Quantum computing, where the quantum computer is used to prepare an algorithm which will run classically \cite{jerbi2023shadows, schreiber2023classical, landman2022classically}, or where a classical machine learning algorithm attempts to replicate the quantum model just using data from the quantum model \cite{huang2021power}.
    
    \item Machine learning, where a large and expensive training run is used to compute a smaller number of weights for a machine learning model. Once these weights are computed the model is often comparatively quite cheap to run \cite{Gavalda94learning,volkovich2014learning,rajgopal2021structure}.
\end{itemize}

With bounded advice classes, we can address these scenarios, we can study how a cryptographic scheme performs against an attacker with `only' an $\mathsf{EXP}$ generated cheat sheet, instead of an unbounded cheat sheet. Formalising bounded advice allows us to exactly extract these ideas, and to ask our key question: \textit{When is advice generated by a given Turing machine useful?} In this context, we can informally say a certain complexity class is a useful advice generator to another complexity class when the former class can be used to create advice strings that allow the latter class to recognize strictly more languages.

In Section \ref{sect: Background} we give common notation and define the class of problems solvable by advice generated by a certain complexity class, e.g. $\mathsf{P/poly^{EXP}}$ is the class of problems solvable by a polynomial-time Turing machine with access to an exponential-time machine. We also review closely connected work on the complexity of advice functions \cite{GAVALDA1995Bounding, Ko85Circuit-Size, Gavalda94learning, Kobler94Complexity-restricted, Kobler1998new, Bshouty1996oracles}. 
This previous work typically studies the functional complexity of generating the advice for a given problem in $\mathsf{P/poly}$.
While our results allow us to push the state of the art in this research direction (we are able to improve the strongest result), we primarily focus on a closely related, but different question: `For a given complexity of advice generator, what problems can be solved?'.

In Section \ref{sect: unary} we prove a connection between bounded advice and unary languages. Namely:

$$\mathsf{P/poly^B = P/poly^{Un(P^B)} = P^{Un(P^B)}},$$

where $\mathsf{Un(B)}$ denotes the set of all unary languages in a class $\mathsf{B}$.

This connection allows us to directly produce results on when given complexity classes generate useful advice in Section \ref{sect: useful?}.

\begin{itemize}
    \item $\mathsf{P \subsetneq P/poly^{EXP}}$
    \item $\mathsf{P \subsetneq P/poly^{NP}} \iff \mathsf{EXP \neq EXP^{NP}}$
    \item $\mathsf{P \subsetneq P/poly^{PSPACE}} \iff \mathsf{EXP \neq EXPSPACE}$
\end{itemize}

We connect our framework to existing quantum classes, $\mathsf{BPP/samp^{BQP}\subseteq P/poly^{Un\left(ZPP^{NP^{BQP}}\right)}}$\footnote{We will later define $\mathsf{BPP/samp^{BQP}}$ as the class of problems which can be solved on a quantum computer or by a $\mathsf{BPP}$ machine with samples from the quantum computer}
and with randomised advice $\mathsf{BPP/samp^{BQP}}\subseteq\mathsf{P/rpoly^{BQP}}$ in Section \ref{sect: useful?}.2. This section also connects bounded quantum advice (where the advice itself is a quantum state) into our framework, proving the bounded-advice equivalent of the well known result, $\mathsf{BQP/qpoly \subseteq QMA/poly}$ \cite{aaronson2014full}:

$$\mathsf{BQP/qpoly^B}\subseteq \mathsf{QMA/poly^{ZPP^{QMA^B}}}.$$

\section{Background}\label{sect: Background}
\subsection{General notation and definitions}
This work makes use of notation or definitions that, while standard, may only be known to readers of specific backgrounds. To aid with readibility by a wide audience we provide a short list of notation, definitions of common complexity classes should be checked in the complexity zoo \cite{ComplexityZoo}.

\begin{itemize} 
\item\textbf{The symbol `$\#$':} It is often useful to join two inputs together to pass them both to a Turing machine, e.g. if we want to calculate $x+y$ for $x=010$ and $y=11$ we will need to pass both $x$ and $y$, but as both are written in binary a simple concatenation makes it unclear where $x$ ends and $y$ begins, $xy=01011$. 
To solve this problem we introduce a special symbol to our alphabet, $\#$, which will be used to simply demarcate where two strings meet, e.g. $x\#y=010\#11$

\item\textbf{Prefixes:} We say \textit{$x$ is a length $n$ prefix of $y$} if $x$ is length $n$ and the first $n$ letters of $y$ are equal to $x$, e.g. $x=100$ is a length 3 prefix of $y=100010$.

\item\textbf{Turing machine standardisation:} Unless otherwise specified a Turing machine is assumed to be a polynomial time deterministic Turing machine. Turing machine is often abbreviated to TM. 
Occasionally the TM will not be deterministic or not polynomial time, this will be specified or be clear from context (i.e. the previous line specified a non-deterministic Turing machine and the following line talks about `this TM').

\item \textbf{Unary:} A unary language is a language such that all elements $x\in L$ consist of a string of 1's, i.e. $\forall x\in L$ $\exists m$ such that $x=1^m$.

The set of all unary languages is written $\mathsf{TALLY}$, all unary languages in some complexity classes $\mathsf{B}$ could be written $\mathsf{TALLY\cap B}$, but we find it is clearer to write $\mathsf{Un(B):=TALLY\cap B}$.

\item \textbf{Sparsity:} A sparse language is a language which has at most $\polyn$ elements of size $n$. 

The class of all sparse languages is written as $\mathsf{SPARSE}$. 

The set of all sparse languages in a complexity class, $\mathsf{B}$, can be written $\mathsf{SPARSE\cap B}$ (the intersection of these two classes) but for notational clarity, we will write it as $\mathsf{SP(B):=SPARSE\cap B}$. This significantly improves readability but may confuse readers familiar with an older form of relativisation notation: $\mathsf{B(C):=B^C}$.
 
\item \textbf{Which exponential time? $\mathsf{E}$ vs $\mathsf{EXP}$.}
There are two standard, but non-equivalent, definitions of exponential time decision problems: $\mathsf{E}$, which is equal to $\mathsf{DTIME(2^{O(n)})}$, and $\mathsf{EXP}$, which is equal to $\bigcup_{\text{all polynomials }p}\mathsf{DTIME(2^{p(n)})}$. 
This work will make it clear when our results apply to only one of these definitions or to both. 
Older works, some cited here, may not conform to this nomenclature.

\item \textbf{$\bm{L(x)}$:} For a language $L$ the function $L(x)$ is equal to 1 if $x \in L$ and 0 otherwise.

\item \textbf{$\bm{1^n}$:} $1^n$ denotes a string of ones of length $n$. $1^4 = 1111$.

\item \textbf{$\mathsf{B-C}$}. For two complexity classes, $\mathsf{B}$ and $\mathsf{C}$, the class $\mathsf{B-C}$ is the set of all the languages in $\mathsf{B}$ that are not also in $\mathsf{C}$. This notation comes from the definition of complexity classes as sets.
\end{itemize}

\subsection{Oracle-machines and double oracles}

The standard definition of an oracular complexity class is given as follows.

\begin{definition}[Oracular classes]
The complexity class of problems solvable by an algorithm in $\mathsf{B}$ with access to an oracle for a language ${L}$ is $\mathsf{B^L}$. This extends to define an oracular complexity class, $\mathsf{C}$ by taking the union:
    $$
    \mathsf{B^C := \bigcup_{L\in C} B^L}
    $$
\end{definition}

Sometimes a single Turing machine will have to make Oracle calls to two different languages, while this can be defined within the existing framework by creating a third language that can answer oracular calls to either language, we find it much simpler to define a double oracle machine.

\begin{definition}[Double oracles]
        For complexity classes $\mathsf{B}$, $\mathsf{C}$ and $\mathsf{D}$, a double oracle machine is a Turing machine, $T$, with oracular access to two problems, $L_0$, and $L_1$. The complexity class is:
    $$
    \mathsf{D^{B,C} = \bigcup_{L_B\in B,L_C\in C} D^{L_B, L_C}}
    $$
\end{definition}

In previous works the notation $\mathsf{L_B\oplus L_C}$ is used for this purpose. While logically equivalent we find the comma notation to be much clearer, additionally it makes sparsity more apparent as either if $\mathsf{B}$ or $\mathsf{C}$ are sparse their set-addition may not be.

\subsection{Bounded and unbounded advice classes} We can now move onto the meat of our definitions: advice classes.

\begin{definition}[Advice classes]
A language $L$ is in the advice class of $\mathsf{B}$, $\mathsf{B/poly}$, if there exists an integer, $d$, and a set of advice strings, $a=\{a_i: |a_i|<d i^d \}_{i\in \mathbb{N}}$, such that the language
$$
L'= \{
x\# a_{|x|}:x\in L
\}
$$
is in $\mathsf{B}$.
\end{definition}

At this point the extension of this definition to bounded advice seems obvious, we just have to specify what it means for a complexity class to `generate' a piece of advice. But it is not clear how a decision algorithm (which outputs a bit) `in' $\mathsf{C}$ generates a piece of advice (a string of symbols).
Fortunately, a lot of heavy lifting has already been done with the definition of  `transducers' \cite{Meduna2000}, which, unlike the standard definition of a Turing machine, gives the entire final state of the tape as the output, where a regular Turing machine only outputs `reject' or `accept'.

Unfortunately, on closer inspection, transducers do not have the properties we want; consider polynomial-sized advice generated by an exponential time machine, which polynomially-sized piece of the exponentially long tape should we take? Or a non-deterministic Turing machine, which path do we accept? What about for even more artificial classes, like $\mathsf{SPARSE}$, which is defined without reference to Turing machines at all! Previous works \cite{GAVALDA1995Bounding, Ko85Circuit-Size, Gavalda94learning, Kobler94Complexity-restricted, Kobler1998new, Bshouty1996oracles} have dealt with these problems by using functional complexity classes. While well-defined, functional complexity classes can have very different properties from their decision-class counterparts \cite{aaronson2023qubit}, they also lack a clear notion of unary languages, which is of key importance to this work.

To combat these issues we choose a definition of bounded advice that uses the tape from a polynomial-time Turing machine with oracular access to another class. This is equivalent to using a $\mathsf{FP^B}$ advice generator for some oracle class $\mathsf{B}$. This definition solves our problem of defining the output tape but inherently transforms the advice generator into a $\mathsf{P^C}$ machine, which makes it difficult to study classes for which $\mathsf{P^C \not = C}$.

\begin{definition}[Polynomial-sized bounded advice classes]
For arbitrary complexity classes, $\mathsf{B}$ and $\mathsf{C}$, a language $L$ is in $\mathsf{B/poly^{C}}$ if there is an infinite set of advice strings, $\{a_n\}_{n\in \mathbb{N}}$, such that:
\begin{itemize}
    \item (Advice Generation) There exists a deterministic polynomial-time `advice generating' Turing machine, T, with oracular access to $\mathsf{C}$ which on input $1^n$ terminates in $\polyn$ time with $a_n$ as the final state of its tape.
    \item (Advice Use) The language 
    $$
    L'= \{
    x\# a_{|x|}:x\in L
    \}
    $$ 
    is in $\mathsf{B}$
\end{itemize}
 
\end{definition}

The notion of \textit{unbounded} advice classes extends beyond classes of the form $\mathsf{B/poly}$; logarithmic advice is possible with $\mathsf{B/log}$, quantum advice, consisting of quantum states is possible with $\mathsf{B/qpoly}$, or randomised advice can be captured with $\mathsf{B/rpoly}$ \cite{aaronson2023qubit}.
Similarly, different flavours of bounded advice are possible; $\mathsf{B/exp^C}$ is the class of problems solvable by $\mathsf{B}$ given advice generated by a Turing machine running for exponentially long, or a quantum advice generator $\mathsf{B/qpoly^C}$, or even a non-deterministic advice generator $\mathsf{B/npoly^C}$. 
Exploring these definitions may allow future papers to work around the `$\mathsf{P}$-oracle problem' (how to study advice generated by a class $\mathsf{C}$ such that $\mathsf{P^C \neq C}$).
For our purposes, we will only need one of the possible extensions of bounded advice.

\begin{definition}[$\mathsf{BQP/qpoly^C}$]
    For a fixed gate set and initial state, $\ket{0}$, a language $L$ is in $\mathsf{BQP/qpoly^C}$ if the following two criteria are true:
    \begin{itemize}
        \item (Advice Generation) There exists a polynomial-sized uniform family of quantum circuits with oracular access to $\mathsf{C}$, $\{Q_n \}_{n\in \mathbb{N}}$ generating a set of advice states, $\{ \ket{\phi_n}=Q_n \ket{0}\}$.
        \item ($\mathsf{BQP}$ Advice Use) There exists a polynomial-time quantum algorithm $A$ such that for all $n$, $A(x,\ket{\phi_n})$  outputs $L(x)$ with probability more than $2/3$ for all $x$ up to length $n$.
    \end{itemize}
\end{definition}

\subsection{Previous work}
We are not the first to consider bounded advice classes although a seemingly minor difference in our framing brings us to a novel set of questions and results.
Much of the previous work on bounded advice has centred on the question `If a given class $\mathsf{A}$ is in $\mathsf{P/poly}$, what is the complexity of generating the advice string?'. This naturally leads to studying the complexity of the `advice function', a function which prints the advice used by the $\mathsf{P/poly}$ algorithm.
The strongest answer to this question was noticed by Köbler and Watanabe~\cite{Kobler1998new} as a corollary to a result by Bshouty et al. \cite{Bshouty1996oracles}, showing that the advice function for any $\mathsf{A\in P/poly}$ is in $\mathsf{FZPP^{NP^A}}$. 

While our bounded advice classes are equivalent to studying the advice function, the reframing leads to a different set of questions.
Instead of asking `for a given problem, what is the advice function to solve this?' we ask `for a given advice function, which problems can this solve?'. 
This naturally leads to our central question of `when are certain advice classes useful?'. It also forces us to standardise our notation and classes, which proves useful to deriving more flexible results. Finally, instead of studying functional classes to specify the advice function, we choose to stay within decision classes. This allows us to use the well-known connection between advice classes and unary languages. 

For clarity, we will restate the state of the art result by Bshouty/Köbler using our terminology.

\begin{theorem}[\cite{Bshouty1996oracles, Kobler1998new}]\label{thm: unbounded advice trans}
    For any set $\mathsf{A \in P/poly}$: $\mathsf{A \in P/poly^{ZPP^{NP^{A}}}}$
\end{theorem}

Reexamining the proof of \cite{Bshouty1996oracles} it is easy to see this generalises to $\mathsf{B/poly}$. We can also use the connection to unary languages we derive later to produce a stronger version of Bshouty et al.'s theorem.

\begin{theorem}\label{cor: unbounded advice trans}
    For any complexity class, $A$, if $\mathsf{A\subseteq C}$ and $\mathsf{A \subseteq B/poly}$ then $\mathsf{A \subseteq B/poly^{Un\left(ZPP^{NP^{B,C}}\right)}}$
\end{theorem}

Many of our other results will rely on the well-known connection between advice classes and unary/sparse languages. As most research in advice functions has been functional languages this connection has scarcely been used. The only connection we are aware of is to a 1985 result by Ko and Schöning \cite{Ko85Circuit-Size}, showing that all sparse sets in $\mathsf{\Sigma_i^p}$ are also in $\mathsf{P/poly}$ with an advice function in $\mathsf{F\Delta^p_{i+1}}$.


\section{Connection to sparse and unary languages}\label{sect: unary}
It is well known that any language in $\mathsf{P/poly}$ is Turing reducible to a unary language. Here we prove an analogous result, that any language in $\mathsf{P/poly^B}$ is Turing reducible to a {unary} language in $\mathsf{P^{B}}$. We show this inclusion is tight, i.e. the set $\mathsf{P/poly^B}$ is exactly the set of languages which are Turing reducible to a unary language in $\mathsf{P^B}$. We further show that all languages that are Turing reducible to sparse languages in $\mathsf{P^B}$ (i.e. in $\mathsf{P^{SP(P^B)}}$) are also in $\mathsf{P/poly^{NP^B}}$.

To prove the main results we will need the following lemmas: 

\begin{lemma}\label{lma: unary in advice}
For any complexity class $\mathsf{B}$:
    $$
    \mathsf{P^{Un({P^B})} \subseteq P/poly^B}.
    $$
\end{lemma}
\begin{proof}
    As $\mathsf{Un(P^B)}$ has at most $n$ elements up to length $n$, the advice can simply be an $n$ length string, $a$, whose $m$'th element is 1 iff $1^m\in L$. Obviously, this string can be generated by $n$ queries to an $\mathsf{P^B}$ language. The advice-receiving Turing machine can then use this advice to simulate oracle calls to the $\mathsf{P^B}$ language.
\end{proof}

\begin{lemma}\label{lma: advice in unary}
For any complexity classes $\mathsf{B}$ and $\mathsf{C}$,
    $$
    \mathsf{C/poly^B \subseteq P^{C, Un({P^B})}}
    $$
\end{lemma}
\begin{proof}
Take any $L \in\mathsf{C/poly^B}$, let $a_n$ be the advice string generated by the advice generator, and let $p(n)$ be the polynomial which bounds the length of the advice string ($\abs{a_n}\leq p(n)$), W.L.O.G. we may assume the advice is a bit string of exactly $p(n)$ bits. We will now describe an oracle in $\mathsf{Un(P^B)}$ that can be used to generate $a_n$ in polynomial time.

Define the language:
$$
L_{advice}=\{
1^{p(n)}1^m|
\text{ the m'th bit of } a_n \text{ is 1 }
\}
$$

By the definition of bounded advice $a_n$ can be generated in polynomial time with a $\mathsf{B}$ oracle, deciding if a particular bit is 1 or 0 is also polynomial time. Therefore $L_{advice}$ is a unary language in $\mathsf{P^B}$.

The language ${L_{advice}}$ can be used to construct $a_n$ one bit at a time. This advice string can then be passed to the $\mathsf{C}$ oracle to simulate the $\mathsf{C/poly^B}$ algorithm, showing that the language $L$ is in $\mathsf{P^{C, Un(P^B)}}$.
\end{proof}

Combining lemma \ref{lma: unary in advice} and lemma \ref{lma: advice in unary} we get the central theorem of this section:

\begin{theorem}\label{thm: advice is unary}
For any complexity class $\mathsf{B}$,
    $$
    \mathsf{P/poly^B = P^{Un({P^B})}}
    $$
\end{theorem}

A simple corollary, $\mathsf{P/poly^B = P/poly^{Un(P^B)}}$, provides us with a key insight: all bounded advice is unary advice. 

The result of Theorem \ref{thm: advice is unary} also holds for advice receivers other than $\mathsf{P}$. The choice of fixing $\mathsf{P}$ as the advice receiver was merely to ensure Turing reductions. Following the steps of the proof for other classes provides the following corollary.

\begin{corollary}
    $$
    \mathsf{BPP/poly^B= BPP^{Un(P^B)}}
    $$
    $$
    \mathsf{NP \cap coNP/poly^B= (NP \cap coNP) ^{Un(P^B)}}
    $$
    $$
    \mathsf{BQP /poly^B} = \mathsf{BQP^{Un(P^B)}}
    $$
\end{corollary}

While the connection to unary languages is tighter, bounded advice can also be connected to sparsity, finding that sparse languages always sit inside some bounded advice class (Theorem \ref{thm: sparse languages are in NP^B advice}), and bounded advice classes are contained in sparse languages (Corollary \ref{cor: advice in sparse}). Theorem \ref{thm: sparse languages are in NP^B advice} is similar to the work of Ko and Schöning \cite{Ko85Circuit-Size} who show that all sparse sets in $\mathsf{\Sigma_i^p}$ are also in $\mathsf{P/poly}$ with an advice function in $\mathsf{F\Delta^p_{i+1}}$. Our result differs as it is both wider-reaching (applying to classes outside the polynomial hierarchy) and showing inclusions in both directions (advice inside sparsity, and sparsity inside advice). 

\begin{theorem}\label{thm: sparse languages are in NP^B advice}
For any complexity class $\mathsf{B}$,
    $$\mathsf{P^{SP(P^B)}\subseteq P/poly^{SP(NP^B)}}$$
\end{theorem}

\begin{proof}
        Let $K$ be a language in $\mathsf{P^{SP(P^B)}}$. Then there exists a sparse language, ${L}$ in $\mathsf{P^B}$ such that $K\in P^L$. As $L$ is sparse there are at most polynomially many strings $x\in L$ of any particular length, $n$. Our strategy will be to find all of these strings with the advice generating TM, to create a list, $A_n=\{x: x\in L, |x|\leq n\}$, and pass $A_n$ as advice. The advice using TM can then simulate an oracle call to $L$ by simply checking if $x$ is on the list.

    It is possible to generate $A_n$ in polynomial time with access to an $\mathsf{NP^B}$ oracle, with access to the (sparse) language: 
    
    $$
    L'_{sparse} = \{
    1^n\#p: \exists x \text{ with } x\in L, \text{ and $p$ is a prefix of }x
    \}.
    $$

    As $L$ is in $\mathsf{P^B}$ existence of an $x\in L$ is in $\mathsf{NP^B}$.

    We can use $L'_{sparse}$ to find a complete list of strings by beginning with an empty list, $A=\{\}$. We begin by using algorithm \ref{Alg:binary search variant} starting from an empty string to find an initial element, $x\in L$, and put this element in $A$. If we try to nïavely reuse algorithm \ref{Alg:binary search variant} to find more elements it may not return a new value of $x\in L$. We must force the algorithm to return a new value of $x$. We can solve this problem by beginning the search from some particular prefix, instead of the empty string.

\begin{figure}[t]
    \centering
    \includegraphics[width=0.8\linewidth]{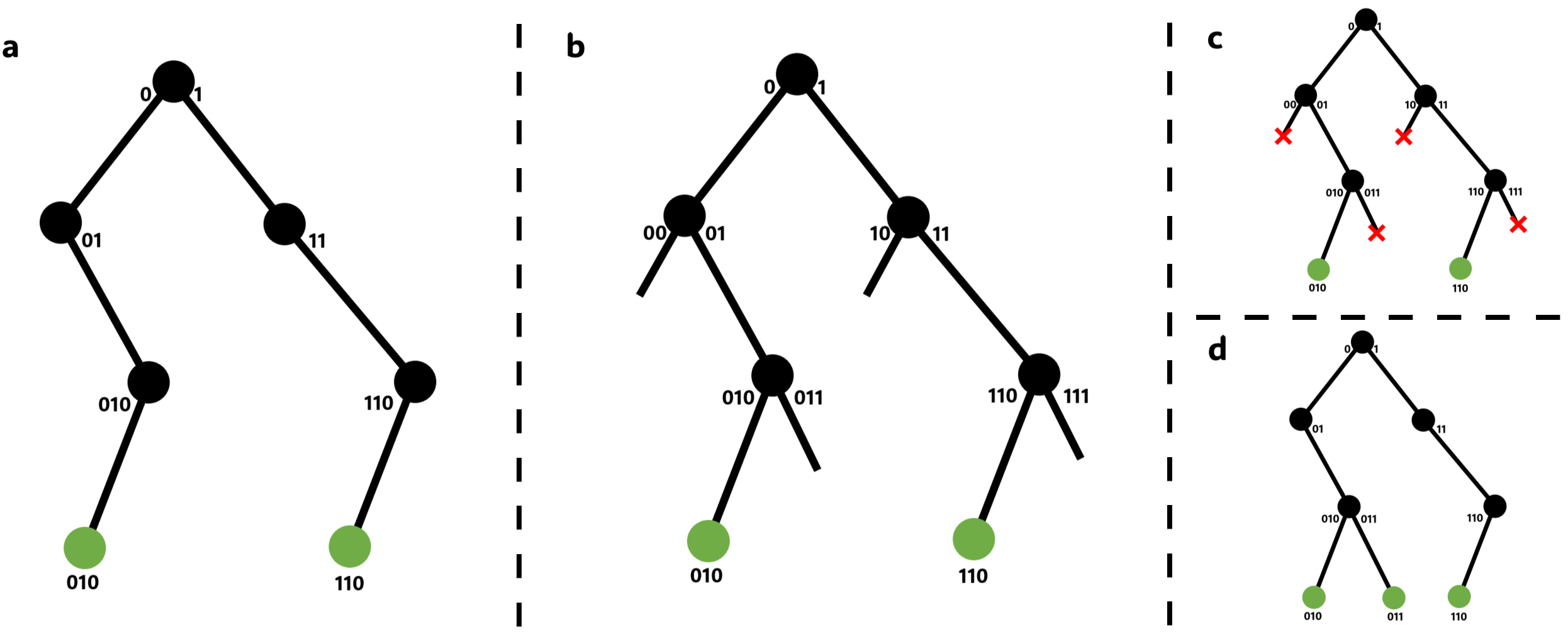}
    \caption{\textbf{Step of the procedure for finding all the elements of a sparse language using $L'$.} We begin with a possibly incomplete list, $A$, consisting of elements of $L$ below to some given size, $n$. In step \textbf{a.} we represent $A$ with a prefix tree, each layer representing a possible prefix of the string. If a string is a prefix for an element in $\mathsf{A}$ its branch continues to a green node. \textbf{b.} We add the possible unexplored children to the tree. We then use algorithm \ref{Alg:binary search variant} starting from each unexplored child prefixes to find possible new elements of $A$. Either one is found and added to $A$ (\textbf{d.}) or none are found, implying there are none left to find and $A$ is complete (\textbf{c.}).}
    \label{fig:trie search}
\end{figure}

    Suppose our list, $A$, has $m$ elements in it. This defines a partially explored trie, as shown in figure \ref{fig:trie search}. At each node there are two children, one child must be an explored path that leads to at least one previously found element of $x\in A$, the other child may be unexplored. By beginning the search with a prefix from one of the unexplored children (highlighted in step b of figure \ref{fig:trie search}) we will find any unfound elements beginning with that prefix. If a new element is found we add it to $A$ and search its unexplored prefixes. If no new elements are found after all prefixes have been explored then we conclude the list $A$ is complete. This list is our advice.

As $A_n$ has at most polynomial elements and the tree is polynomially deep, finding each element takes at most polynomial time. Creating the advice string $A$ therefore takes polynomial time with an $\mathsf{NP^B}$ oracle, as $L'_{sparse}$ is also sparse, we have shown the result:

    $$\mathsf{P^{SP(P^B)}\subseteq P/poly^{SP(NP^B)}}.
    $$

    \begin{algorithm}
    \label{Alg:binary search variant}
        \caption{Binary search variant to find $a_n$ from $L_{sparse}$}
        \textbf{Input:} An integer $n$, a starting prefix, $p$

        \textbf{Output:} The string $a'$
        \begin{algorithmic}[1]
            
            \State $a' \gets $ ``$p$" \Comment{Initializing $a'$ to the given prefix}
            \State $advice\_complete \gets $ \textbf{False} \Comment{Initializing $advice\_complete$ to \textbf{False}}
            \While{$advice\_complete$ is \textbf{False}}
                \If{$1^n\#a'0\in L$} \Comment{Find a symbol to append to the advice string}
                \State $a' \gets a'0 $
                \ElsIf{$1^n\#a'1\in L$}
                \State $a' \gets a'1 $
                \Else
                \State $advice\_complete \gets $ \textbf{True} \Comment{If no next symbol can be found, then the string is complete}
                \EndIf
            \EndWhile
            
        \end{algorithmic}
    \end{algorithm}
\end{proof}

Theorem \ref{thm: sparse languages are in NP^B advice} shows sparse languages are in a form of bounded advice, the converse is also possible as a simple corollary of Theorem \ref{thm: advice is unary} as $\mathsf{Un(B)\subseteq Sp(B)}$.

\begin{corollary}\label{cor: advice in sparse}
    $\mathsf{P/poly^B\subseteq P^{SP(P^B)}}$
\end{corollary}


\section{When is bounded advice useful?}\label{sect: useful?}
This section will address this paper's fundamental question: \textit{``When is advice generated by a given complexity class useful?"} i.e. when advice from one class increases the amount of problems that can be solved by another class. In the first subsection, we will characterise when some major complexity classes generate advice useful to a polynomial time machine (i.e. when $\mathsf{P/poly^A}$ is not simply equal to $\mathsf{P}$). In the second subsection, we look at when quantum complexity classes are useful advice generators, this allows us to study a class of proposed uses of quantum computing that prepare classical algorithms. 
We connect $\mathsf{BPP/rpoly^{BQP}}$ to the class of languages that can be decided with samples from a quantum computer through $\mathsf{BPP/samp}$\cite{huang2021power}. We derive a result connecting quantum bounded polynomial advice (i.e. a quantum state) to non-quantum bounded polynomial advice (a classical bitstring).

\subsection{Useful classical advice}
As shown in the previous section $\mathsf{P/poly^B = P^{Un(P^B)}}$, thus bounded advice is useful if and only if unary languages are useful oracles. Much is known about the conditions for the existence of various unary languages, thus, this connection opens a variety of choices for grounding the hardness of various bounded advice classes. We provide a number of these results for the most famous complexity classes.

\begin{theorem}\label{cor: conditions for useful advice, EXP case}
    $\mathsf{P \neq P/poly^{EXP}}$
\end{theorem}

\begin{theorem}\label{cor: conditions for useful advice, PSPACE case}
    $\mathsf{P\neq P/poly^{PSPACE}}$ if and only if $\mathsf{EXPSPACE \neq EXP}$.
\end{theorem}

\begin{theorem}\label{cor: conditions for useful advice, NP case}
    $\mathsf{P\neq P/poly^{NP}}$ if and only if $\mathsf{EXP^{NP} \neq EXP}$.
\end{theorem}

\begin{theorem}\label{thm: conditions for useful advice, BQP case}
    $\mathsf{BPP\neq BPP/poly^{BQP}}$ if and only if $\mathsf{BPEXP\neq BQEXP}$
\end{theorem}

Many other results in unary languages can be used to prove results in advice classes. Such as showing the polynomial hierarchy gives useful advice if the exponential hierarchy doesn't collapse. However, for brevity, we have provided only these 4.

\begin{proof}[Proof of Theorem \ref{cor: conditions for useful advice, EXP case}]
    We prove this result via showing $\mathsf{P\neq P^{Un(EXP)}}$.
    By the time hierarchy theorem, there exists an $L$ such that,
    $$
    L\in \mathsf{DTIME(2^{n^3}) - DTIME(2^{n^2})}.
    $$
    From $L$ we define the unary language:
    $${L_{unary} = \{1^x : x\in L
    \}}$$
    ${L_{unary}}$ cannot be decided in $o(2^{log(n)^2})$, a quasipolynomial, thus $\mathsf{L_{unary}}\notin \mathsf{P}$.

    Given a string from $L_{unary}$, $1^x$, calculating $x$ takes $x$ steps,
    therefore ${L_{unary}}$ can be decided by a deterministic Turing machine in $O(2^{log(n)^3})$ time. As $\mathsf{DTIME(2^{log(n)^3}) \subseteq EXP}$, $L_{unary}\in \mathsf{EXP}$. Therefore we have demonstrated there exists a unary language in $\mathsf{EXP-P}$, by theorem \ref{thm: advice is unary} this implies $\mathsf{P/poly^{EXP} \neq P}$
\end{proof}

\begin{proof}[Proof of Theorem \ref{cor: conditions for useful advice, PSPACE case}]
By Theorem \ref{thm: advice is unary} $\mathsf{P/poly^{PSPACE}=P^{Un(P^{PSPACE})}}$, as $\mathsf{P^{PSPACE}= PSPACE}$~\cite{ComplexityZoo} we derive the equality: $\mathsf{P/poly^{PSPACE}=P^{Un({PSPACE})}}$. 
  If $\mathsf{EXPSPACE = EXP}$ there are no unary languages in $\mathsf{PSPACE-P}$~\cite{hartmanis1983onSparse}, proving the `only if' direction. For the other direction we notice that all unary languages in $\mathsf{PSPACE}$ are in $\mathsf{P^{Un(P^{PSPACE})}}$, therefore if  $\mathsf{P^{Un(P^{PSPACE})}=P}$ there are no unary languages in $\mathsf{PSPACE}$ and $\mathsf{EXPSPACE = EXP}$~\cite{hartmanis1983onSparse}.
\end{proof}

\begin{proof}[Proof of Theorem \ref{cor: conditions for useful advice, NP case}]
There are unary languages in $\mathsf{P^{NP}-P}$ if and only $\mathsf{EXP^{NP}\neq EXP}$.
Therefore $\mathsf{P/poly^{NP}}$ $\mathsf{\neq P}$ if and only if $\mathsf{EXP^{NP}\neq EXP}$.
\end{proof}

\begin{proof}[Proof of Theorem \ref{thm: conditions for useful advice, BQP case}]
There are unary languages in $\mathsf{BQP-BPP}$ if and only $\mathsf{BPEXP\neq BQEXP}$, by a simple extension of the arguments in \cite{hartmanis1983onSparse}.
Therefore $\mathsf{P/poly^{BQP}}$ $\mathsf{\neq P}$ if and only if $\mathsf{BQEXP\neq BPEXP}$.
\end{proof}

\subsection{Useful Quantum advice}
This subsection studies bounded quantum advice and when it is useful. 
First, we examine classical advice strings generated by quantum computers, this case contains many of the proposed algorithms to use quantum computers to prepare algorithms for classical machines \cite{jerbi2023shadows, huang2021power, cerezo2023does, schreiber2023classical}. 
One prominent example, using samples from some problem to produce an algorithm that can solve instances of the same problem, is formalised in the class $\mathsf{BPP/samp}$, if the samples are produced from a quantum computer, we denote it as the class $\mathsf{BPP/samp^{BQP}}$ \cite{huang2021power}. 
While it may be expected that $\mathsf{BPP/samp^{BQP}}$ is contained in $\mathsf{P/poly^{BQP}}$ this neglects the randomness in the sample selection, and we will show it is instead contained in $\mathsf{P/rpoly^{BQP}}$. We also show a derandomised result, that $\mathsf{BPP/samp^{BQP}}$ is contained in $\mathsf{P/poly^{Un\left(ZPP^{NP^{BQP}}\right)}}$.
Second, we study the properties of bounded advice when the advice is itself a quantum state: $\mathsf{BQP/qpoly^B}$, finding it is contained in a classical bounded advice state $\mathsf{QMA\cap coQMA/poly^{Un(ZPP^{QMA^B})}}$.

\begin{definition}[$\mathsf{BPP/samp}$ \cite{huang2021power}]
    A language $L$ is in $\mathsf{BPP/samp}$ if there exists probabilistic polynomial-time Turing machines $M$ and $D$ with the following properties: On input $1^n$, $D$ produces output distribution $\mathcal{D}_n$. 
    $M$ takes an input of size $n$ along with `samples' from $L$, $\mathcal{T}= \{(x_i, L(x_i))\}_{i=1}^{poly(n)}$, where $x_i$ is sampled from $\mathcal{D}_n$.
    $M$ outputs 1 with probability greater than $2/3$ if ${x\in L}$ , and less than $1/3$ if $x\notin L$, where the probability taken is over the randomness in both sample selection and random coins.
\end{definition}

The original definition of $\mathsf{BPP/samp}$ \cite{huang2021power} is not explicit which domain the $2/3$ failure probability applies to, it could be that for most sets of samples the machine $M$ must function `according to the rules of $\mathsf{BPP}$' (for all points there is a $2/3$ probability of failure over the coin flips), or it could be that for each point, a $2/3$ probability of failure applies over both the set of samples and the set of coin flips. While the former appears to be a much tighter restriction (the majority of all sets work on all points for most sets of coin flips), fortunately, these two definitions are equivalent via a simple boosting-and-majority-vote argument.
Similarly, it is clear that the use of $\mathsf{BPP}$ in the definition of $\mathsf{BPP/samp}$ is superfluous as if there is randomness over the samples, this randomness is sufficient to use as randomness to simulate coin flips. If a given $\mathsf{BPP/samp}$ algorithm requires $m$ random coins, we can ask for extra samples and use the random inputs $x$ as random coins\footnote{The random distribution of samples, $\mathcal{D}_n$, might not be uniform, infact there is no requirement for it to be non-deterministic. Fortunately, we can always append $m/\abs{x}$ uniform random samples on the end of our set of samples to use as random coin flips.}.

\begin{theorem}
    $\mathsf{BPP/samp=P/samp}$
\end{theorem}

For our purposes, we are interested in exactly the restriction of $\mathsf{BPP/samp}$ to samples that can be produced by a particular machine (i.e. a quantum machine/$\mathsf{BQP}$). We define $\mathsf{BPP/samp^B}$ equivalently to $\mathsf{BPP/samp}$ but requiring that the labelling problem ($L(x)$) is in the complexity class $\mathsf{B}$. Fortunately, this is simply $\mathsf{BPP/samp\cap B}$. 

\begin{lemma}
    $
    \mathsf{BPP/samp^B=BPP/samp\cap B}
    $
\end{lemma}

From this definition, the following result is immediately clear:

\begin{theorem}\label{thm: samp}
$\mathsf{BPP/samp^{BQP}\subseteq BPP/rpoly^{BQP}}$
\end{theorem}

As noted by Watrous \cite{watrous2008quantum} randomised advice can be used to give random coins to a probabilistic algorithm,
which extends to bounded advice to show $\mathsf{BPP/rpoly^{BQP}= P/rpoly^{BQP}} $. This connection improves Theorem \ref{thm: samp} to the following corollary.

\begin{corollary}\label{cor: samp}
$\mathsf{BPP/samp^{BQP}\subseteq P/rpoly^{BQP}}$
\end{corollary}

As noted above, the definition of $\mathsf{BPP/samp}$ seems to require randomised advice. Fortunately, we can use Theorem \ref{thm: samp} to derandomise this advice, connecting $\mathsf{BPP/samp^{BQP}}$ to a deterministic bounded advice class.

\begin{theorem}
    $$
    \mathsf{BPP/samp^{BQP}\subseteq 
    P/poly^{Un\left(ZPP^{NP^{BQP}}\right)}
    }
    $$
\end{theorem}
\begin{proof}
    As $\mathsf{BPP/samp^{BQP}\subseteq BQP}$ and $\mathsf{BPP/samp^{BQP}\subseteq P/poly}$, applying Theorem \ref{thm: samp} produces the desired result.
\end{proof}

Similar these techniques bound other `quantum preparation classes', such as  $\text{CSIM}_\text{QE}$~\cite{cerezo2023does}.

We can now turn our attention to quantum states given as advice, the standard equivalent result in unbounded advice classes for connecting quantum advice to advice is $\mathsf{BQP/qpoly\subseteq QMA\cap coQMA/poly}$ \cite{aaronson2014full}, we show a close analogue exists for bounded advice classes:

\begin{theorem}
$$\mathsf{BQP/qpoly^A \subseteq QMA\cap coQMA/poly^{ZPP^{QMA^A}}}$$
\end{theorem}
\begin{proof}
    Aaronson \cite{aaronson2014full} and Drucker showed:
    $$
    \mathsf{BQP/qpoly\subseteq QMA \cap coQMA/poly}
    $$
    As $\mathsf{BQP/qpoly^A\subseteq BQP^A}$ we can apply Theorem \ref{thm: unbounded advice trans} to derive:
    $$
    \mathsf{BQP/qpoly^A \subseteq QMA\cap coQMA/poly^{Un(ZPP^{QMA^{QMA\cap coQMA, BQP^A}})}}
    $$
    This result can be significantly simplified, first, we note that $\mathsf{QMA \cap coQMA}$ is low for $\mathsf{QMA}$
    $$
    \mathsf{QMA^{QMA \cap coQMA, BQP^A}\subseteq QMA^{BQP^A}}
    $$
    Then, we use $\mathsf{QMA^{BQP^A}=QMA^{A}}$ to derive the result
    $$\mathsf{BQP/qpoly^A \subseteq QMA\cap coQMA/poly^{Un(ZPP^{QMA^A})}}$$
\end{proof}

\section*{Acknowledgements}
VD and SCM acknowledge the support by the project NEASQC funded from the European Union’s Horizon 2020 research and innovation programme (grant agreement No 951821). VD and SCM also acknowledge partial funding by an unrestricted gift from Google Quantum AI. CG was supported by the Dutch Research Council (NWO/OCW), as part of the Quantum Software Consortium programme (project number 024.003.037).
This work was supported by the European Union’s Horizon Europe program through the ERC ERC CoG BeMAIQuantum (Grant No. 101124342).
This work was also supported by the Dutch National Growth Fund (NGF), as part of the Quantum Delta NL programme.

\bibliographystyle{plain}
\bibliography{biblo.bib}

\end{document}